\documentclass{vldb}
\usepackage{algorithm,algorithmic,subfigure}
\usepackage{xfrac}
\RequirePackage{fix-cm}

\newcommand*\samethanks[1][\value{footnote}]{\footnotemark[#1]}

\title{On Demand Memory Specialization for Distributed Graph Databases}

\numberofauthors{3}

\author{
\alignauthor
Xavier Martinez-Palau\thanks{DAMA-UPC, Computer Architecture
Department, Technical University of Catalonia, Barcelona}\\
       \email{xmartine@ac.upc.edu}
\alignauthor
David Dominguez-Sal\samethanks\\
       \email{ddomings@ac.upc.edu}
\alignauthor Reza Akbarinia\thanks{INRIA - LIRMM, Montpellier, France}\\
       \email{reza.akbarinia@inria.fr}
\and
\alignauthor Patrick Valduriez\samethanks\\
       \email{patrick.valduriez@inria.fr}
\alignauthor Josep Llu\'is Larriba-Pey\samethanks[1]\\
       \email{larri@ac.upc.edu}
}

\newtheorem{theorem}{Theorem}[section]
\newtheorem{lemma}[theorem]{Lemma}
\newtheorem{corollary}[theorem]{Corollary}

\begin{document}
\maketitle

\begin{abstract} 

In this paper, we propose the DN-tree that is a data structure to
build lossy summaries of the frequent data access patterns of the
queries in a distributed graph data management system. 
These compact representations allow us an efficient communication
of the data structure in distributed systems.
We exploit this data structure
with a new \textit{Dynamic
Data Partitioning} strategy (DYDAP) that assigns the portions of the
graph according to historical data access patterns, and guarantees
a small network communication and a computational load balance in
distributed graph queries.
This method is able to adapt dynamically to new workloads and evolve
when the query distribution changes. 
Our experiments show that DYDAP
yields a throughput up to an order of magnitude higher than previous methods
based on cache specialization, in a variety of scenarios, and the average
response time of the system is divided by two.

\end{abstract}

\section{Introduction}

Graph databases have become quite popular over the last years. Social networks,
bibliographic relations or metabolical pathways are some examples of datasets
naturally expressed as graphs. In these scenarios, most computations that are
performed against these datasets can be expressed as graph queries. For
example, the home page of a social network, which shows the new posts of the
users' friends, can be implemented as a two hop traversal that navigates from
the user to his friends and then to the published messages. Also, large
analytical  operations like finding the most influential users of a social
network, can be implemented by computing the central users of the network with
a sequence of breadth first search traversals of the
graph~\cite{brandes2001faster}.

The size of graph datasets is usually very large, and grows daily.
For instance, in March~2011, Twitter users generated about a billion tweets per
week~\cite{twitter}, which corresponds to more than $1500$~tweets per second.
As the size of the datasets increases, they become more and more difficult to
manage on a single machine and graph partitioning becomes necessary. In
addition, some environments require many queries answered per unit of time,
making it also difficult for a single computer to cope with them. Graph
partitioning is a common technique used to maintain data locality in
distributed systems. Some systems partition the graph
statically~\cite{yang2012towards} but this requires computing the min cut of
the graph, which is an NP-hard problem, and thus can be very expensive for some
datasets. As a result, most of the current solutions, such as
Pregel~\cite{malewicz2010pregel} or ParallelGDB~\cite{barguno2011parallelgdb},
ignore this problem and partition the vertices by hashing their identifiers.



In this paper, we propose a dynamic approach, that we call DYDAP, for dynamic
data partitioning, which is able to summarize the graph accesses through a
compact data structure, that we call the DN-tree, for density tree. When a
query is launched, the system analyzes the DN-tree contents and partitions the
graph by taking into account the data access patterns of previous queries. Our
approach does not compute the min cut of the whole graph, but only of a small
graph induced by the DN-tree. Therefore, it ensures graph access locality while
being scalable.

Our first contribution is the design of the DYDAP system for graph computation
that divides the graph into small portions that have good locality properties
based on the analysis of previous queries. Our approach distinguishes two
levels: the secondary storage and the memory manager. The secondary storage is
an independent persistent storage for the graph structure and its attributes.
The memory manager specializes the memory of each node\footnote{Through all this
paper, we use \emph{node} to refer to a computer, and \emph{vertex} to refer to
a graph entity.} in the cluster and takes advantage of this cache
specialization, allowing for faster query execution. The rest of our
contributions are focused on the efficiency of the memory manager.

Our second contribution is the proposal and analysis of the DN-tree, a data
structure that captures the sequence of data accesses. Since the amount of
information can be very large, the DN-tree performs efficient compression, as
analyzed in Section~\ref{sec:theoretical}. The memory manager uses the DN-tree
to analyze the data access patterns of ongoing queries, and is able to compute
 a new partitioning of the data at regular time intervals. Each new
partitioning is based on the execution patterns of the queries that the system
has executed until that point. This allows the system to adapt to varying
workloads and improve the throughput of new queries. Other systems in the
literature using similar partitioning schemes are static, meaning that the
distribution of data is decided before starting the execution of queries and
does not change during
execution~\cite{barguno2011parallelgdb,curino2010schism}. These systems do not
adapt to the nature of the incoming queries, and while they are optimized for
some queries, if the incoming queries change, they are no longer optimal.

Our third and final contribution is to balance the load and network
communication between nodes. The memory manager uses the information stored in
the DN-tree data structure to compute a partitioning of the data that balances
the load in each node while simultaneously balancing and minimizing the amount
of network communication. As the experiments described in
Section~\ref{sec:experiments} show, this translates into an increased
throughput of the queries executed and a decreased average response time, which
is cut in half.

The rest of the paper is organized as follows: Section~\ref{sec:definition}
formally defines the problem. Section~\ref{sec:overview} introduces a general
vision of the DYDAP system. Section~\ref{sec:DN-tree} describes the DN-tree
data structure. Section~\ref{sec:manager} describes the partitioning strategy
of our system. Section~\ref{sec:theoretical} analyzes the asymptotic size of
the DN-tree structure and provides bounds on the error due to its compression.
Section~\ref{sec:experiments} describes the experiments and results obtained.
Section~\ref{sec:related} discusses related work. Section~\ref{sec:conclusions}
concludes the paper.

\section{Problem Definition}
\label{sec:definition}



We assume a shared nothing cluster where data is divided into chunks of
equal size, named \emph{extents}, each of them with a unique identifier. An
extent stores arbitrary data, and it is up to the database management system to
decide how to store data in the extents. Due to the shared nothing architecture
used, each node accesses a subset of the extents, and each pair of subsets is
disjoint. Thus, there is a mapping between extents and nodes, that
assigns a unique node to each extent. We call this mapping \emph{distribution
function}.


The execution of a query in the system is modeled as a Bulk Synchronous
Processing (BSP) computation similar to that of the BFS algorithm
in~\cite{yoo2005scalable}. In this model, there are points in time where
network communication is allowed, and intervals, named \emph{phases}, where
each node does computations on the data assigned to it according to the
distribution function. The duration of each phase is determined by the slowest
node, because all nodes need to finish their computations before network
communication begins.

Our model assumes that the graph database can be updated, adding new vertices
or edges. It also assumes that there are trends in the content searched by the
queries and this pattern does not radically change over time. For example, a
web server generating dynamic content issues similar queries although each
query accesses a different set of the data. The model also assumes that there
is no a priori knowledge about the queries that the system executes.


Our objective is to find an optimal distribution function, i.e. one that meets
the following two objectives:

\vspace{-0.15cm}
\begin{enumerate}
\itemsep=-.3em
 \item Minimize network communication.
 \item Balance the load among the nodes during all computation phases.
\end{enumerate}
\vspace{-0.15cm}

The second objective focuses on load balancing. This balancing is twofold. On
the one hand, the total workload for each node should ideally be the same; on
the other hand, all nodes should always be performing useful computations.
These two sides of load balancing are illustrated in the following example.

\subsubsection*{Example}

Consider the following simplified example. The system consists of two nodes,
labeled $n_1$ and $n_2$. The database consists of four extents, 
$\{e_1, e_2, e_3, e_4\}$. We also assume that accessing the information on an
extent and doing the required computations consumes one unit of time, and that 
exchanging through the network the information generated when accessing each
extent costs one unit. There are two computational phases and one communication
step in between.

A query $q$ accesses all extents, but $e_1$ and $e_2$ need to be accessed
strictly before $e_3$ and $e_4$. A possible distribution function is one that
maps $e_1$ and $e_2$ to $n_1$, and $e_3$ and $e_4$ to $n_2$. In the first
phase, $n_1$ accesses $e_1$ and $e_2$, using 2 units of time while $n_2$ is
idle. During the network communication phase, $n_1$ communicates the necessary
information to $n_2$, using $2$ units of network cost. During the second phase,
$n_2$ accesses $e_3$ and $e_4$ while $n_1$ is idle. The time consumed in the
second phase is $2$ units, and the total cost is $4$ units of time plus $2$
units of network cost.

A second possible distribution function assigns $e_1$ and $e_3$ to $n_1$, and
$e_2$ and $e_4$ to $n_2$. In the first phase, $n_1$ accesses $e_1$ and $n_2$
accesses $e_2$ in parallel, and the time consumed is one unit. Both nodes
communicate through the network, using $2$ units of network cost, and similarly
the second phase needs one unit of time, with $n_1$ accessing $e_3$ and $n_2$
accessing $e_4$. The total time is $2$ units of time plus $2$ units of network
cost.

We note that both distribution function minimize network communication.
However, only the second one balances the load. The first distribution function
only balances the load globally, as both nodes have the same global amount of
workload. The important difference is that the second distribution function
parallelizes the computations.

\subsection{Problem Formalization}

In this section, we formalize the problem that we solve to obtain a distribution
function, the \emph{Horizontal Multiconstraint Partitioning Problem}
(HMCPP)~\cite{karypis1998multilevel}. Let $G=(V,E)$ be a graph with $|V|=r$. A
partitioning of $G$ in $k$ parts is a map $P: V \longrightarrow [0,k-1]$. Each
vertex $v \in V$ belongs to partition $P(v)$. Given an edge $e \in E$ that
joins two vertices $u,v \in V$, we say that $P$ cuts $e$ if $P(u) \neq P(v)$.

Each edge $e \in E$ has an associated scalar weight $w_e$, and each vertex $v
\in V$ has an associated weight vector $w^v$ of size $m$. Each component of
this vector is a balancing constraint, so there are $m$ constraints. Also, we
assume, without loss of generality, that each constraint adds to $1$ when added
along the whole graph, $\sum_{v \in V} {w_i^v} = 1$, for $0 \le i < m$.  The
values of these constraints have to be balanced with respect to the partitions.
Given $P$, a partitioning of $G$, we define the imbalance of constraint $i$ as
$$l_i = k \max_{0 \le j < k} \left( \sum_{v \in P^{-1}(j)} w_i^v \right)$$
for $0 \le i < m$. The imbalance is $k$ times the maximum value of the sum of
the weights in each partition.

An optimal solution occurs when all constraints are balanced, i.e. $\sum {w_i}
= \frac{1}{k}$ and thus $l_i$ is minimized, $l_i = 1$. If a solution is not
optimal, then $l_i > 1$.

We define $c$ as a vector of size $m$, with $c_i \ge 1$. Each $c_i$ represents
the maximum imbalance allowed for constraint $i$. Also, we define the
\emph{edge cut} as the sum of the weights of the edges cut by $P$, i.e. $\sum
w_e$ for each vertex $e=(u,v)$ such that $P(u) \neq P(v)$.


The HMCPP is defined as: finding a partitioning $P$ of $G$ in $k$ parts that
minimizes the edge cut, while $l_i \le c_i$ for $0 \le i < m$. The solution
that satisfies the HCMPP condition fulfills the previously presented objectives
of network communication and load balance.

This problem is a well known NP-hard problem arising in several situations, for
example VLSI circuit design~\cite{andreev2006balanced} or detection of cliques
in social, pathological and biological networks~\cite{patkar2003efficient}. It
has been extensively studied, and is solved by approximate algorithms very
fast. One software that partitions graphs is
METIS~\cite{karypis1998multilevel}, which is the one we use.

\section{DYDAP Overview}
\label{sec:overview}

\begin{figure}[t]
\centering
\epsfig{file=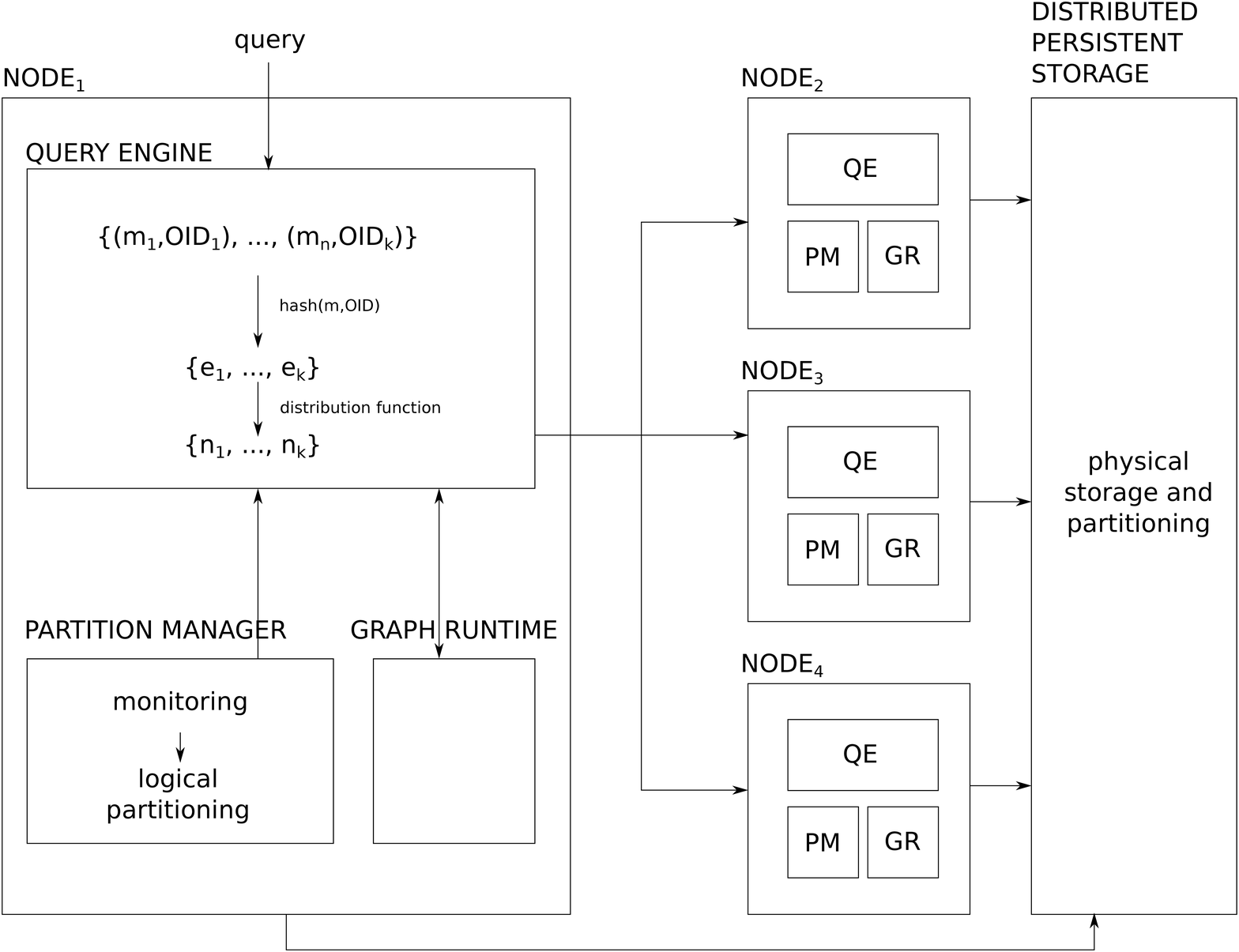, width=0.8\linewidth}
\caption{ Overview of the distributed system. }
\label{fig:esquema}
\end{figure}

Figure~\ref{fig:esquema} shows the architecture of DYDAP, which logically
separates the secondary storage (physical partitioning) from the 
memory manager in each node of the system.

The \textbf{physical storage} is used to store the data used by the system. Any
physical storage able to store extents of fixed size indexed by a unique
identifier can be used by the memory manager, which gives a useful decoupling
between persistent storage and computation. 

The memory manager distributes the data stored into the main memory
of each node in an optimal way, according to the data analyzed during query
execution. Each of the nodes that form the memory manager, four in
Figure~\ref{fig:esquema}, contains three subsystems: the graph runtime, the
query engine, and the partition manager.

The \textbf{graph runtime} is a graph database that is able to store locally a
graph and execute graph operations. This database stores data in extents, and
has to be modified to use the distributed storage as its permanent storage
instead of the local hard disk. Additionally, it has to record accesses to
extents and report them to the partition manager.

The \textbf{partition manager} monitors the accesses to the extents and stores
the information in a DN-tree. At periodic intervals of time, the partition
manager starts a repartition task in background and updates the distribution
function. The first part of this repartition task is to gather and aggregate
the DN-trees built in each node and merge them in a DN-tree of the full system,
described in Section~\ref{sec:DN-tree}. The second part is the generation of
the new distribution function as described in Section~\ref{sec:manager}. During
the repartition process, the queries can be executed using the last
distribution function until a new one is computed. So, this process does not
delay or interfere significantly with incoming queries.


Every time a new query enters the system, the \textbf{query engine} analyzes
the data that the query accesses, and distributes it across the system, using
the distribution function provided by the partition manager. 

\section{The DN-tree Data Structure}
\label{sec:DN-tree}

The partition manager captures the sequence of data accesses in the graph
database and uses this information to improve its partitioning. The sequence of
data accesses is recorded as a matrix of transitions, $M$, which in turn can be
also viewed as a graph, where cell $M_{a,b}$ counts the number of times that
extent $b$ has been accessed after $a$. If there are $m$ extents in the
database, storing this information requires an $m \times m$ matrix, called $M$.
This matrix can be used later to detect data access patterns that happen often
in the system. For example, if the value of $M_{a,b}$ is large, then we will be
able to improve the data access locality of the system by assigning extents $a$
and $b$ to the same node, because they are accessed often together. However,
since the matrix $M$ can be very large, we propose an alternative data
structure that approximates $M$, the density tree, or \emph{DN-tree}. The
DN-tree can be seen as a temperature map of the sequences of accesses to
extents, which may be hot (often accessed) or cold (seldom accessed).

\subsection{DN-tree Formalization}

The DN-tree data structure consists of a rooted tree where each non-leaf
vertex has exactly four children. Each vertex of the tree is associated with a
subset of the matrix. The root vertex is associated with the whole matrix, and
each of the four children of a vertex is associated to a spatial partitioning
of the space in four equal quadrants. If a vertex of the tree monitors rows
between $a$ and $b$, and columns between $c$ and $d$, we write this range as
$[a,b] \times [c,d]$. For such a vertex, each of its four children monitors
accesses to extents in the ranges

\begin{small}
\begin{itemize}
 \item $\left[a,\sfrac{(a+b)}{2}\right] \times \left[c,\sfrac{(c+d)}{2}\right]$
 \item $\left[a,\sfrac{(a+b)}{2}\right] \times \left(\sfrac{(c+d)}{2},d\right]$
 \item $\left(\sfrac{(a+b)}{2},b\right] \times \left[c,\sfrac{(c+d)}{2}\right]$
 \item $\left(\sfrac{(a+b)}{2},b\right] \times \left(\sfrac{(c+d)}{2},d\right]$
\end{itemize}
\end{small}

The root vertex is special and only contains pointers to its four children. The
range of the matrix associated with the root vertex is $[0,m] \times [0,m]$ and
thus the four subareas associated to its children are:

\begin{small}
\begin{itemize}
 \item $[0,m/2] \times [0,m/2]$
 \item $[0,m/2] \times (m/2,m]$
 \item $(m/2,m] \times [0,m/2]$
 \item $(m/2,m] \times (m/2,m]$
\end{itemize}
\end{small}

Each vertex different from the root vertex stores an integer that accounts
for the number of accesses that the region associated has had.

{\flushleft \textbf{Insertion}:} First, we describe constructively the data
structure. The insertion operations are inspired by
quadtrees~\cite{finkel1974quad}, as described in Algorithm~\ref{alg:update}.
Initially, the DN-tree only contains the root vertex. 

If extent $b$ is accessed after extent $a$, the value of the root child
associated with the area containing $(a,b)$ is incremented in one unit. This
happens every time accesses to different extents are done. This counter may
reach a threshold value $\tau(i)=t \cdot k^i$, where $t$ and $k$ are constants
(see Section~\ref{sec:theoretical} for a discussion on these values.) The
parameter of $\tau$ is accounted by \emph{vertex.level} in
Algorithm~\ref{alg:update} (the distance between the vertex and the root of the
tree). For the root children, $\tau(1)=t k$. Once the counter reaches the
threshold, the counter is not modified anymore and the DN-tree grows to record
this area with more detail. Four new vertices are created, which are rooted on
the saturated counter. Thereafter, the new vertices monitor the range of values
of the saturated value, with the new counters initialized to zero. This is the
source for the lossy nature of these data structure, as there is no way to know
the distribution to the lower levels of the saturated counters. The method
$selectChild$ returns a pointer to the child associated with the subarea where
$(a,b)$ belongs to.

\algsetup{indent=2em,linenosize=\small}
\begin{algorithm}[!htb]
 \caption{update(vertex, a, b)}
 \fontsize{8}{8}\selectfont
 \begin{algorithmic}[1]
  \REQUIRE \emph{vertex}: current vertex of the tree.
  \REQUIRE \emph{a}: previously accessed extent.
  \REQUIRE \emph{b}: last accessed extent.

  \IF{vertex.value == $\tau(\mathrm{vertex.level})$}
    \STATE let \emph{child} = vertex.selectChild(a, b)
    \STATE update(child, a, b)
  \ELSE
    \STATE vertex.value = vertex.value + 1;
  \ENDIF
  \RETURN
 \end{algorithmic}
 \label{alg:update}
\end{algorithm}

{\flushleft \textbf{Query}:} The DN-tree is used to compress a matrix $M$.
Since the compression used is lossy, it is not possible to recover $M$, but an
approximation $\hat{M}$. The DN-tree allows us to compute $\hat{M}(i,j)$, which
is the value in row~$i$ and column~$j$ of $\hat{M}$, and it is an approximation
of $M(i,j)$. The procedure to compute $\hat{M}(i,j)$ is a recursive traversal
of the DN-tree from the root node to the leafs as shown in
Algorithm~\ref{alg:query}, where \emph{value} is initially $0$. For each
vertex, we accumulate the counter corresponding to the subarea of $(i,j)$, and,
in case it is not a leaf node, the counter is weighted by the number of
accesses in that vertex. This procedure guarantees a reduced error as discussed
in Section~\ref{sec:theoretical}. In this algorithm, \emph{vertex.childSum}
stores the sum of the values of the four children of \emph{vertex}.

\algsetup{indent=2em,linenosize=\small}
\begin{algorithm}[!htb]
 \caption{density(vertex, a, b, value)}
 \fontsize{8}{8}\selectfont
 \begin{algorithmic}[1]
  \REQUIRE \emph{vertex}: current vertex of the tree.
  \REQUIRE \emph{a}: row extent.
  \REQUIRE \emph{b}: column extent.
  \REQUIRE \emph{value}: reference where the return value is stored.

  \STATE value = value + vertex.value;
  \IF{vertex is not a leaf node}
    \STATE let \emph{w} = vertex.selectChild(a, b)
    \STATE value = value * w.value / vertex.childSum
    \STATE density(w, a, b, value);
  \ENDIF
  \RETURN
 \end{algorithmic}
 \label{alg:query}
\end{algorithm}

The DN-tree takes advantage of the query access patterns and the spatial
locality of data structures in a graph database. In real databases, most
transitions between extents seldom occur or will not occur at all. Thus, many
sections of the matrix $M$ have a low density of accesses and can be
effectively compressed by the DN-tree. The DN-tree works under the assumption
that extents with close identifiers have similar access patterns. The DN-tree
groups consecutive extents, forming bigger ranges of data.


\begin{figure}[t]
\centering
\subfigure[Density of $M$]{
  \epsfig{file=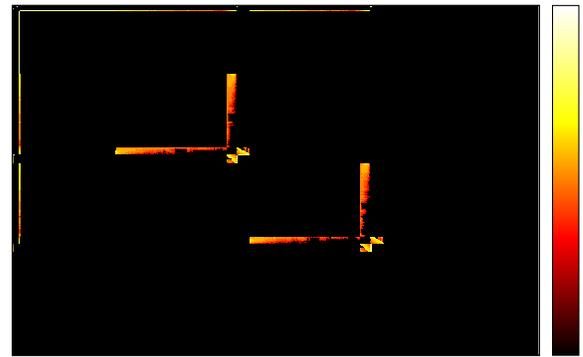, width=0.90\linewidth}
  \label{fig:compressed_a}
}
\subfigure[Density of $\hat{M}$]{
  \epsfig{file=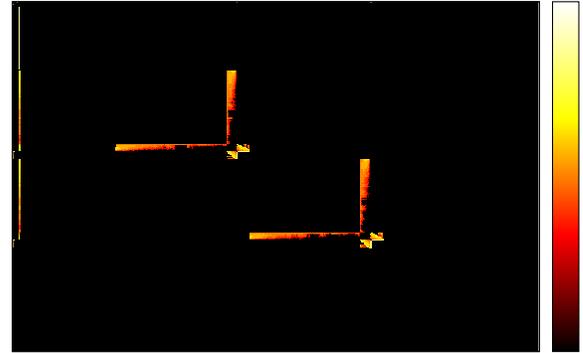, width=0.90\linewidth}
  \label{fig:compressed_b}
}
\caption{ Density of $M$ and $\hat{M}$. }
\label{fig:compressed}
\end{figure}

The areas with more extent accesses are more important, and thus the DN-tree
data structure gives more detail than in areas with few or no accesses. As an
example, Figure~\ref{fig:compressed} shows the associated matrix obtained with
a DN-tree along with the original uncompressed matrix that it approximates.
The data is obtained by running $110$~breadth first search queries on a graph
starting at a random vertex. The graph has about $45$~thousand vertices and
$600$~thousand edges, and the resulting database has $512$~extents The matrix
uses one MiB of memory while the DN-tree needs $138$~kiB, a compression ratio
of $13\%$. During the execution of the queries, there are about $340$~million
extent accesses. The figure shows that the original matrix and its
approximation are very similar, and also that the cold regions dominate over
the hot regions.

\subsection{Distributed aggregation of DN-trees}

In the DYDAP system, each node monitors its own activity, so the DN-tree
only records the local accesses in each node. In order to aggregate the
DN-trees, we use a tree-based scheme of communication to avoid
bottlenecks. Each node receives the DN-tree from its two children, processes
it, and passes it to its parent. The root node computes the next distribution
function and passes it back to its children.

In order to send the DN-tree using the network, the DN-tree is serialized as
the preorder list of the tree described. Along with the value of each vertex,
there is a boolean, named marker, that stores \texttt{false} if the
corresponding vertex has no children, and \texttt{true} otherwise.

Each node combines its own DN-tree with the ones it receives. The root node
has the combination of all partial DN-trees. Since the communication scheme
has a tree-form, in a cluster with $n$ nodes there are $n - 1$ data exchanges.

Algorithm~\ref{alg:join} shows how two serialized DN-trees are combined. This
algorithm does one sequential read of each serialized DN-tree. The algorithm
recursively walks over the arrays containing the serialized DN-trees, treating
three cases: the first two cases correspond to when one of the DN-trees has a
vertex in a branch that the other one does not. In this case, the value is
directly copied from the DN-tree with the branch, and the next recursive call
is executed if necessary. The third case corresponds to both DN-trees having
the vertex in the branch. In this case, the values of the vertices are added
and the next recursive call is executed as needed. We notice that during this
process, some vertices in the resulting DN-tree may have a counter value higher
than the threshold. Since the resulting DN-tree is only used by the partition
manager to extract information and is not updated, this is not a source of
problems.

\algsetup{indent=2em,linenosize=\small}
\begin{algorithm}[!htb]
 \caption{join(src1, index1, src2, index2, ret)}\label{alg:join}
 \fontsize{8}{8}\selectfont
 \begin{algorithmic}[1]
  \REQUIRE \emph{src1}: array storing the first DN-tree.
  \REQUIRE \emph{index1}: integer that indexes \texttt{src1}. If its value
is $-1$, \texttt{src1} is assumed to be empty.
  \REQUIRE \emph{src2}: array storing the second DN-tree.
  \REQUIRE \emph{index2}: integer that indexes \texttt{src2}. If its value
is $-1$, \texttt{src2} is assumed to be empty.
  \REQUIRE \emph{ret}: array, initially empty, to store the result.
  \ENSURE The returned DN-tree \texttt{ret} is a compressed representation of
the sum of the two matrices represented by the two DN-tree parameters.
  \medskip

  \STATE let \emph{mark1, mark2} be boolean values.
  \FOR{i from 1 to 4}
    \IF{index1 == -1}
      \STATE ret.push\_back(src2[index2++]);
      \STATE mark2 = src2[index2++];
      \STATE ret.push\_back(mark2);
      \IF{k != 0}
        \STATE join(src1, index1, src2, index2, ret);
      \ENDIF
    \ELSIF{index2 == -1}
      \STATE ret.push\_back(src1[index1++]);
      \STATE mark1 = src1[index1++];
      \STATE ret.push\_back(mark1);
      \IF{k != 0}
        \STATE join(src1, index1, src2, index2, ret);
      \ENDIF
    \ELSE
      \STATE ret.push\_back(src1[index1] + src2[index2]);
      \STATE index1++; index2++;
      \STATE mark1 = (src1[index1++] != 0);
      \STATE mark2 = (src2[index2++] != 0);
      \STATE ret.push\_back(mark1 or mark2);
      \IF{mark1 and mark2}
        \STATE join(src1, index1, src2, index2, ret);
      \ELSIF{mark1}
        \STATE join(src1, index1, src2, -1, ret);
      \ELSIF{mark2}
        \STATE join(src1, -1, src2, index2, ret);
      \ENDIF
    \ENDIF
  \ENDFOR
  \RETURN
 \end{algorithmic}
\end{algorithm}

\section{Partition Manager}
\label{sec:manager}

The partition manager provides a dynamic distribution function that changes
over time depending on the queries executed and the load of each node. In this
section, we describe the partitioning method used by the partition manager.
Specifically, we describe how the Horizontal Multi-constraint Partitioning
Problem maps to our system and how objectives $1$ and $2$, namely minimization
of network communication and load balancing, described in
Section~\ref{sec:definition}, are achieved.

\subsection{Partitioning Approach}
\label{sub:detail}

In the graph $G$, each vertex is associated with an extent of the database. The
edges are described by the adjacency matrix $\hat{M}$ obtained from the DN-tree
as described in Section~\ref{sec:DN-tree}. Additionally, the value of $k$ is
the number of nodes in our cluster. Thus, each of the partitions correspond to
a node in the cluster.

Objective~$1$ requires that the network communication is minimized. Given an
edge that joins vertices $i$ and $j$, its associated weight, $M_{ij}$,
is the number of times that extent $j$ has been accessed after $i$, it
corresponds to the number of network messages exchanged when the nodes
responsible for $i$ and $j$ are not the same. HMCPP minimizes the edge cut,
which translates into the minimization of the amount of network communication.
This is because edges between vertices belonging to the same partition
involve no network communication, whereas edges between vertices belonging to
different partitions do.

In order to comply with objective~$2$, we define several weight constraints.
The first constraint, $w_0$, is defined as $w_0^v = 1$ for all $v \in V$. With
this constraint, each subset is of a similar size and the total load of each
node is balanced. As described in Section~\ref{sec:definition}, the workload
has to be parallelized. In order to accomplish this, we define more
constraints, one for each data structure in the database whose extents can be
accessed concurrently. Assuming there are $n_\mathrm{DS}$ different data
structures, and naming $\mathrm{DS}_i$ the set of extents belonging to the data
structure number~$i$, we define $n_\mathrm{DS}$ additional constraints, for $1
\le i \le n_\mathrm{DS}$:
$$
w_i^v=
 \left\{
  \begin{array}{ll}
   1 & \mbox{if } v \in \mathrm{DS}_i \\
   0 & \mbox{if } v \notin \mathrm{DS}_i
  \end{array}
 \right.
$$

The solution to HMCPP gives a logical partitioning that assigns a partition
to each vertex. In DYDAP, this corresponds to assigning one node to each
extent. This map from extents to nodes is the distribution function.

\subsection{Example}
\label{sub:matexample}

In this section, we illustrate how the partition manager works. We consider a
system with $2$~nodes and $4$~extents, numbered $0$ to $3$, and show how the
DN-tree is constructed and used to partition the data.

Consider that during execution, the four extents of the system are accessed
according to the following sequence:

\begin{small}
\{1, 2, 1, 3, \textbf{0, 1}, 3, 1, \textbf{0, 1}, 0, 2,
$\mathrm{1}^{\diamond}$, 3, 1, 3, 0, 2, 1, 0, 2, 1, 3, 0, 3, \textbf{0, 1},
\textbf{0, 1}, 3, 1, 3, 1, 2, \textbf{0, 1}, 3, 1, 3, 1, 2, 1, 2, 1\}
\end{small}

This sequence is stored in a matrix $M$ as follows. Notice that we have marked
the $5$ transitions from $0$ to $1$ in the input and the corresponding matrix
entry $M_{0,1}$.

\begin{small}
$$M = \left(\begin{matrix}
0 & \textbf{5} & 3 & 1 \\
4 &         0  & 4 & 9 \\
1 &         6  & 0 & 0 \\
4 &         6  & 0 & 0 \\
\end{matrix}\right)$$
\end{small}

We consider the sequence of accesses and see how the DN-tree is generated.
Level~$0$ of the DN-tree is the root node. It points to the four vertices in
level~$1$, represented in a $2 \times 2$ matrix, with all four values set to
$0$. The four quadrant correspond to different sets of extents, as shown in
Figure~\ref{fig:matrix_a}. We set the parameters $t=4$, $k=1$ and start updating
the DN-tree with each extent access. The first four accesses are $1$, $2$, $1$
and $3$, so the corresponding values are increased. The DN-tree is now as in
Figure~\ref{fig:matrix_a}.

\begin{figure}[t]
\centering
\subfigure[]{
  \epsfig{file=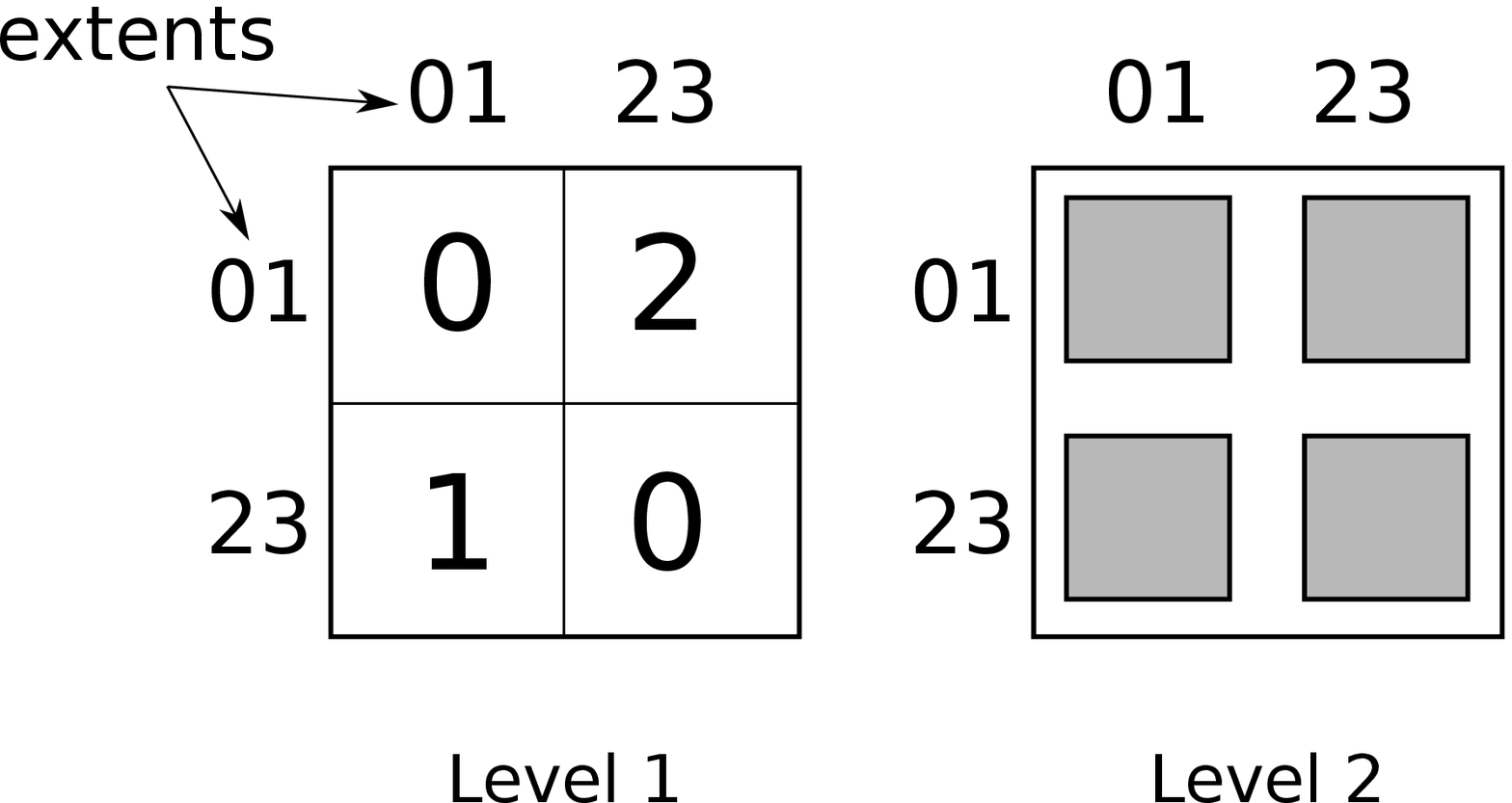, width=0.51\linewidth}
  \label{fig:matrix_a}
}
\subfigure[]{
  \epsfig{file=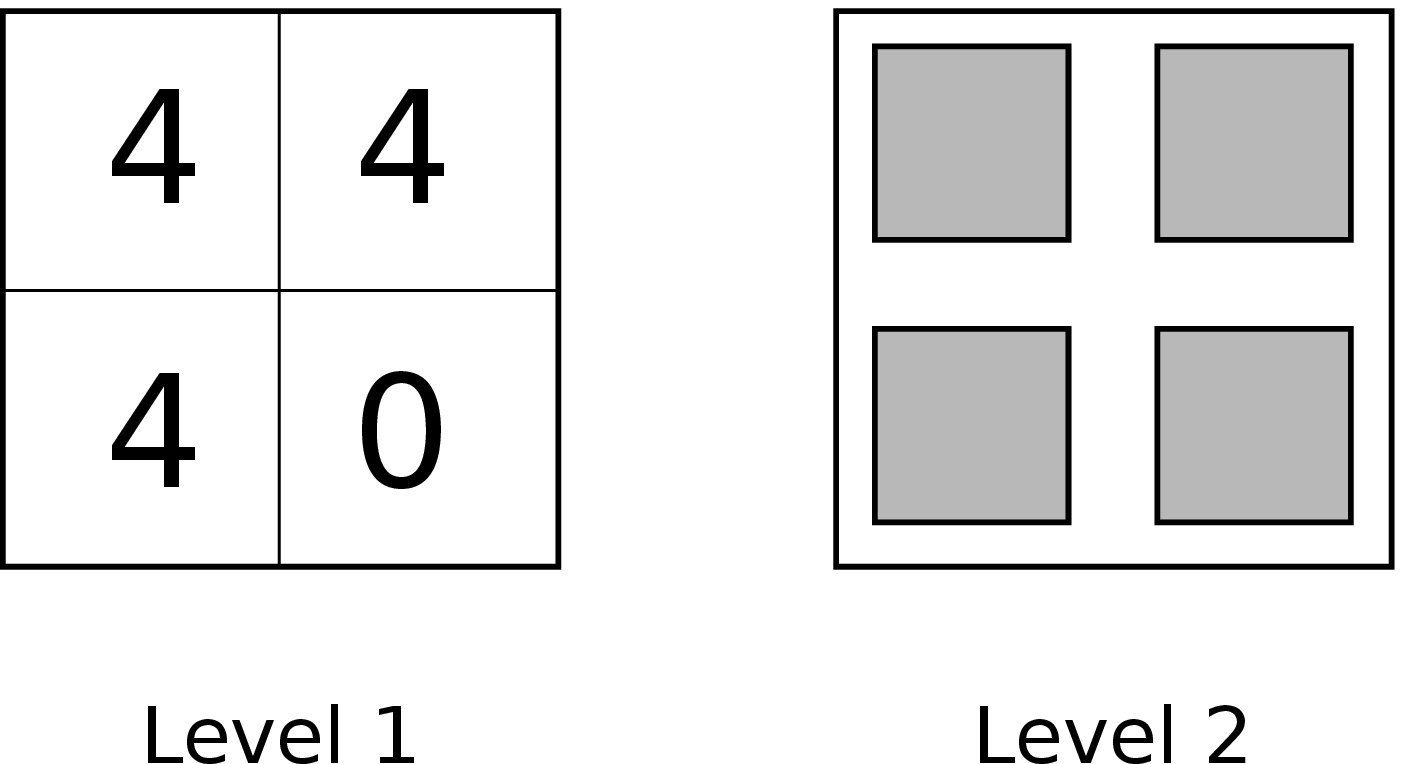, width=0.40\linewidth}
  \label{fig:matrix_b}
}
\subfigure[]{
  \epsfig{file=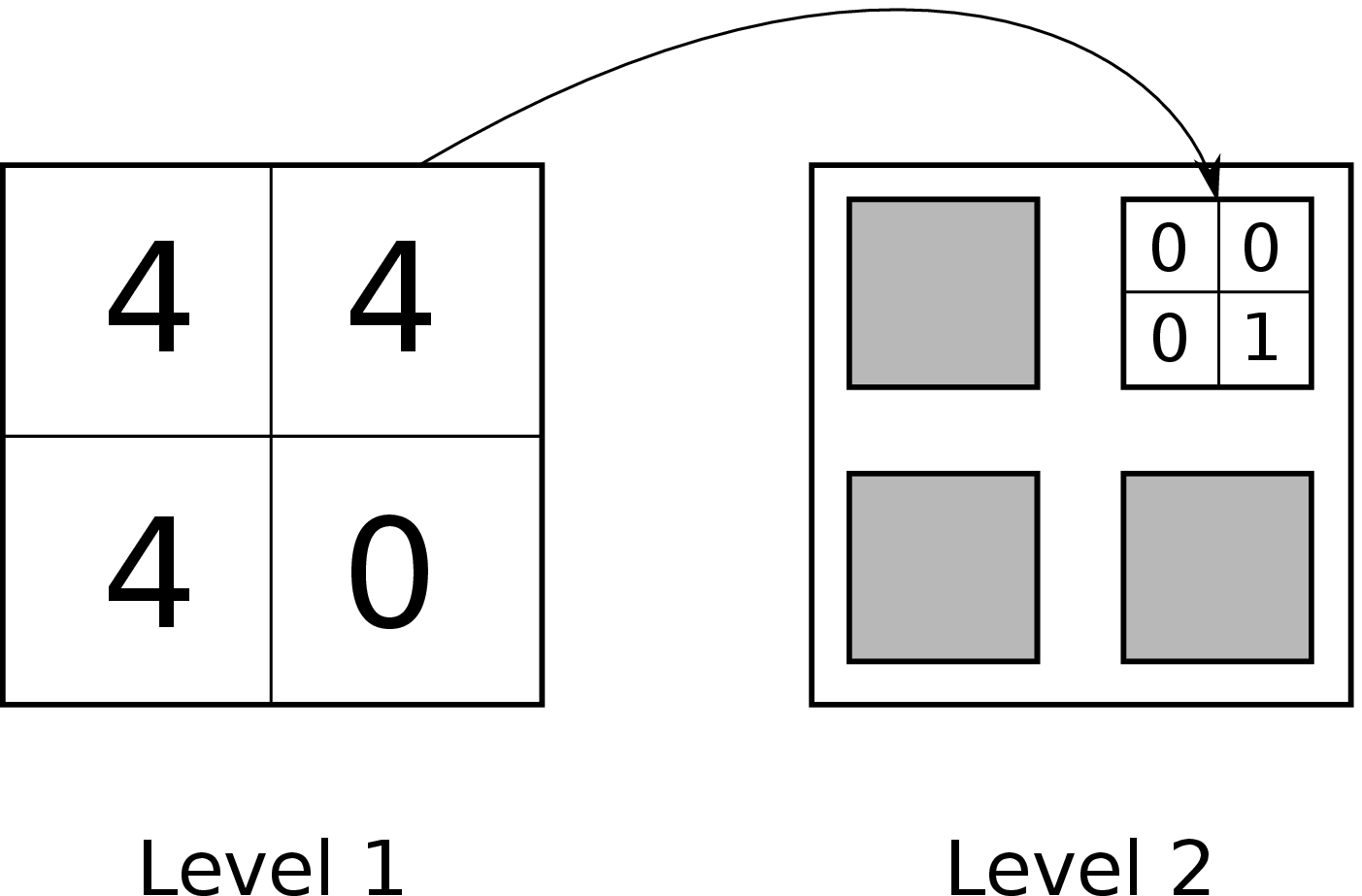, width=0.40\linewidth}
  \label{fig:matrix_c}
}
\subfigure[]{
  \epsfig{file=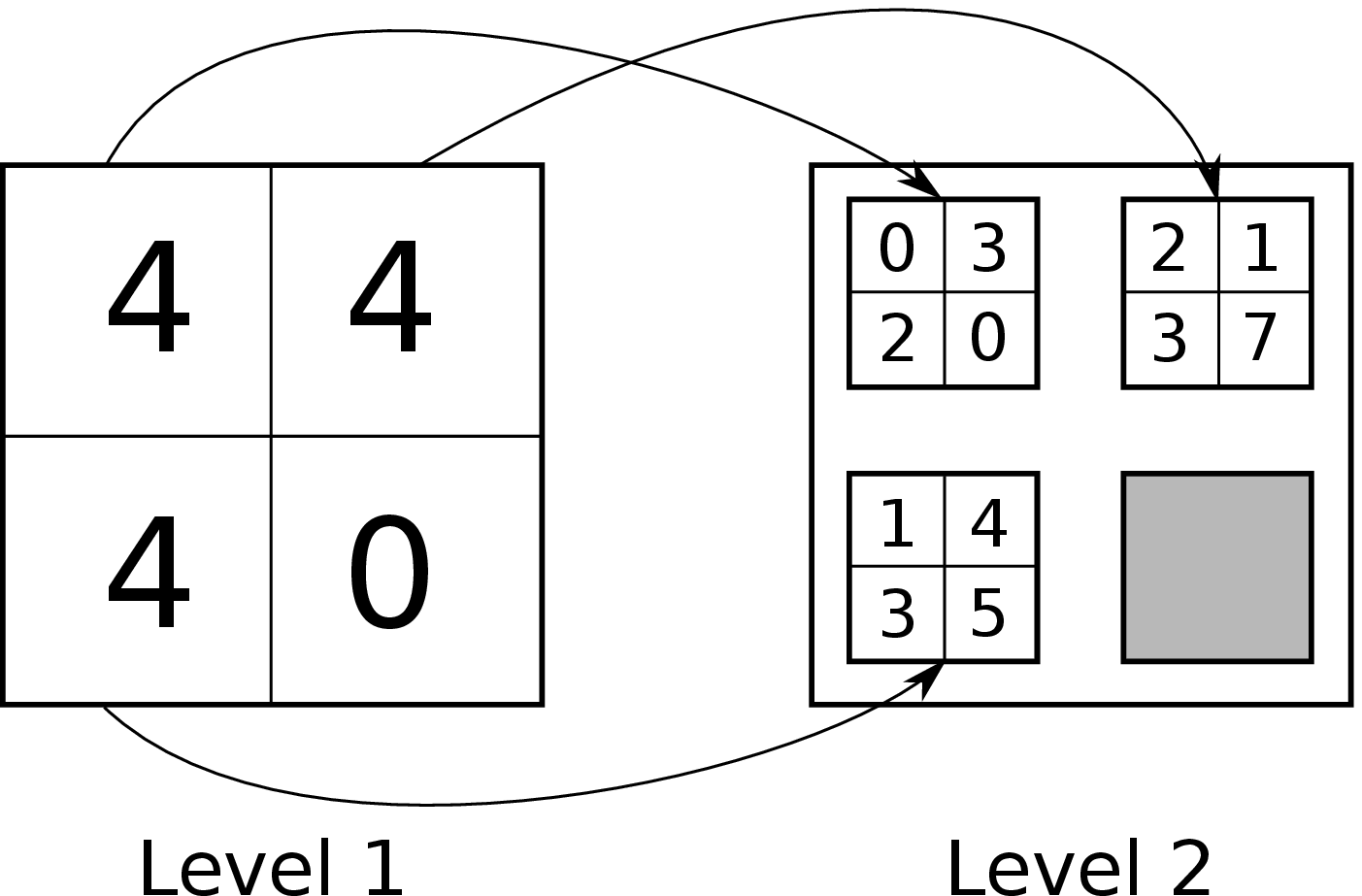, width=0.40\linewidth}
  \label{fig:matrix_d}
}
\caption{ Evolution of the DN-tree. }
\label{fig:matrix}
\end{figure}

After the first $13$ accesses, marked in the list with the diamond $\diamond$,
the DN-tree is as shown in Figure~\ref{fig:matrix_b}. The next transition is
from extent $1$ to $3$, which corresponds to position $(0,1)$, which has a
value of $4$, the threshold value. Now, four new vertices are recursively
generated and initialized. They will be modified during the following accesses.
This is shown in Figure~\ref{fig:matrix_c}. After all $44$ accesses, three of
the four vertices at level~$1$ point to four other vertices each. This is
represented in Figure~\ref{fig:matrix_d}. 

The values $4$ and larger at level~$2$ do not span new vertices because the
extent range associated with these positions has size one and cannot be
partitioned any further. In this example, the DN-tree is very small but not all
branches have full depth. In real scenarios, branches have different depths.

An approximation of the real matrix $M$ is calculated with
Algorithm~\ref{alg:query}. The resulting approximation is

\begin{small}
$$\hat{M} = \left(\begin{matrix}
0 & 5 & 3 & 1 \\
4 & 0 & 4 & 9 \\
1 & 5 & 0 & 0 \\
4 & 7 & 0 & 0 \\
\end{matrix}\right)$$
\end{small}

Comparing the original matrix $M$ with the approximation $\hat{M}$, we see that
only two values are different, and by one unit. In this small example, both $M$
and the DN-tree consist of $16$ integer values, and thus there is no
compression. However, with large matrices the compression ratio is very large,
as analyzed in Section~\ref{sec:theoretical}.

Now, we calculate the optimal partitioning using both $M$ and its approximation
$\hat{M}$. We consider the undirected graphs given by $M + M^T$ and $\hat{M} +
\hat{M}^T$ as their adjacency matrix. As an example, the graph given by $M$ is
represented in Figure~\ref{fig:undirected}. Its adjacency matrix is:

\begin{small}
$$M + M^T = \left(\begin{matrix}
 0 &  9 &  4 &  5 \\
 9 &  0 & 10 & 15 \\
 4 & 10 &  0 &  0 \\
 5 & 15 &  0 &  0 \\
\end{matrix}\right)$$
\end{small}

\begin{figure}[t]
\centering
\epsfig{file=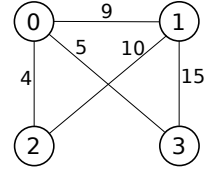, width=0.3\linewidth}
\caption{ Graph of network communication. }
\label{fig:undirected}
\end{figure}

We partition the four vertices of the graph into two subsets of equal size.
This leads to the different possible partitionings. Table~\ref{tab:partitions}
shows these partitionings along with their edge cut. Note that the edge cuts
calculated with $\hat{M}$ are close to the real edge cuts obtained using $M$.
The partitioning that minimizes the edge cut is the same in both cases. This
partitioning minimizes the network communication between the two nodes. Note
that the partitioning foundwith $\hat{M}$ is the same as the optimal
partitioning that could be derived from $M$.


\begin{table}[t]
\centering
\begin{small}
  \begin{tabular}{ | c | c | c | c | }
    \hline
      Partition 1 & Partition 2 & edge cut ($M$) & edge cut ($\hat{M}$) \\
    \hline
      \{0, 1\} & \{2, 3\} & 34 & 34 \\
    \hline
      \{0, 2\} & \{1, 3\} & 24 & 23 \\
    \hline
      \{0, 3\} & \{1, 2\} & 28 & 29 \\
    \hline
  \end{tabular}
\end{small}
  \caption{ Partitionings and edge cuts with $M$ and $\hat{M}$. }
  \label{tab:partitions}
\end{table}

\section{DN-tree Analysis}
\label{sec:theoretical}

In this section, we present a theoretical analysis of the properties of the
DN-tree described in Section~\ref{sec:DN-tree}. Specifically, we analyze the
asymptotic behavior of the memory needed to store the data structure, as well
as an analysis of the error associated with the lossy nature of the
compression.

\subsection{Size Analysis}
\label{sub:size}

In this section, we give the asymptotic behavior of the size of the DN-tree as
a function of the number of extent accesses, $N$. Assuming that the database
has $m$ extents, the uncompressed matrix would need $O(m^2)$ to be stored. We
note that the asymptotic analysis assumes that both $m$ and $N$ tend to
infinity. The number of accesses increases with time, but the number of extents
also increases as new vertices and edges are added to the working dataset.


As described in Section~\ref{sec:DN-tree}, each level has an associated
threshold. The first level, $0$, has a threshold value of $t$, and at each new
level, the threshold is multiplied by a constant factor $k > 1$. Thus, at level
$r$, the threshold value is $tk^r$. We assume that each vertex of the DN-tree
uses $1$ unit of memory to store its associated value.

One last assumption is that at each level, the accesses are distributed along
the four children vertices following the same probability distribution: If we
number the child nodes from $0$ to $3$, then the accesses are distributed
following probabilities $p_i$, $0 \le i < 4$, with $\sum_{i=0}^{3}p_i=1$. This
means that the contents follow an R-MAT model~\cite{ChakrabartiZF04}.

We define a function $T(n,r)$ with two parameters: $n$, the number of extent
accesses, and $r$, a level of the DN-tree. $T(n,r)$ gives the amount of memory
necessary to store $n$ extent accesses starting at level $r$. Thus, the amount
of memory used by a DN-tree storing $N$ extent accesses is $T(N,0)$.

The two parameters of $T$ allow to define a recurrence over $T$. Given $n$ and
$r$, the base case is that $n$ is less than the threshold at level $r$, $n \le
tk^r$. If this is the case, the vertex at this level has no children and the
amount of memory needed is $1$. If $n > tk^r$, then $tk^r$ is stored at level
$r$, and the rest of the accesses, $n - tk^r$ are distributed to the four
children at level $r+1$ following the probabilities $p_i$. Thus, the recurrence
is defined as:

$$
T(n,r)=
 \left\{
  \begin{array}{ll}
   1 & \mbox{if } n \le t k^r \\
   1 + \sum_{i=0}^{3} T\left( p_i \left(n - tk^r \right) , r + 1 \right) &
   \mbox{if } n > t k^r
  \end{array}
 \right.
$$

In order to solve this recurrence, we first simplify it using the following
lemma.

\begin{lemma}
  $T(N,0) = T'(N)$, where
  $$
  T'(n)=
  \left\{
   \begin{array}{ll}
    1 & \mbox{if } n \le t \\
    1 + \sum_{i=0}^{3} T' \left( p_i \frac{n-t}{k} \right) &
    \mbox{if } n > t
   \end{array}
  \right.
  $$
\end{lemma}

\begin{proof}
 We prove the following stronger statement, of which the lemma is a particular
case:

$$T(n, r) = T'\left( \frac{n}{k^r} \right)$$

The statement is proven using reverse induction. For every value of $n$, there
is a value $r_0 = \lceil \log_k \frac{n}{t} \rceil$, such that for all
$r>r_0$, the statement is trivially true: $T(n, r) = T'\left( \frac{n}{k^r}
\right) = 1$, since in both cases the base case condition is met.

Now, we prove that $T(n, r) = T'\left( \frac{n}{k^r} \right)$ implies
$T(n, r-1) = T'\left( \frac{n}{k^{r-1}} \right)$:

\begin{eqnarray*}
T'\left( \frac{n}{k^{r-1}} \right) &=& 1 + \sum_{i=0}^{3}
T' \left( \frac{p_i}{k} \left(\frac{n}{k^{r-1}} - t \right) \right) \\
&=& 1 + \sum_{i=0}^{3}
T' \left( \frac{p_i}{k^r} \left(n - tk^{r-1}\right) \right) \\
&=& 1 + \sum_{i=0}^{3}
T \left( p_i \left(n - tk^{r-1}\right), r \right) \\
&=& T(n, r-1)
\end{eqnarray*}

Thus, for every $n$, $T(n,0) = T'(n)$.
\end{proof}

The new recurrence is interpreted as follows. Instead of increasing the
threshold by a factor $k$, we maintain the threshold constant but reduce $n$ by
the same factor in the recurrent call. With this, parameter $r$ is a constant
and the new recurrence has only one parameter, which makes it easier to solve.
Specifically, we observe that it can be solved using the Akra-Bazzi method (see
Appendix~\ref{app:akra-bazzi_theorem}.) This method provides an analytical
solution with a parameter $s$. Although we are not able to compute the exact
value of $s$, we are able to use numerical methods and give a few bounds for
the worst case in the following theorems.

\begin{theorem}
\label{the:DN-tree_size}
 The asymptotic space used by a DN-tree with parameters $k$ and $t$, is
$T(N) \in \Theta \left( N^s \right)$, where $N$ is the number of extent
accesses and $s$ is the solution to $$\sum_{i=0}^{3} p_i^s = k^s$$
\end{theorem}
\begin{proof}
We apply Akra-Bazzi theorem.
Our recurrence has $g(n)=1$, $c=3$, $a_i=1$, $b_i = \frac{p_i}{k}$ and
$h_i(n)=-\frac{p_it}{k}$. It is easy to check that all conditions are met.
In our case,
$$\int_{1}^{x} \frac{g(u)}{u^{s+1}}\,\mathrm{d}u =
  \int_{1}^{x} \frac{1}{u^{s+1}}\,\mathrm{d}u =
  s^{-1} \left( 1 - x^{-s} \right)$$
Thus, $T(N) \in \Theta \left( N^s \right)$, $\sum_{i=0}^{3} p_i^s = k^s$
\end{proof}

We note that the asymptotic behavior does not depend on the value of $t$. Out
of the two parameters needed to specify a DN-tree, $k$ and $t$, only $t$ has an
effect on the asymptotic size of the DN-tree.

\begin{theorem}
\label{lem:bounds}
  The asymptotic space used by a DN-tree with parameters $k \ge 1$ and $t$, is
$T(N) \in \Theta \left( N^s \right)$, where $N$ is the number of extent
accesses and $s$ is a real number between $0$ and $\frac{\log 4}{\log 4k}$.
\end{theorem}

\begin{proof}
We define $f(s) = \sum_{i=0}^{3} p_i^s$. Using Lagrange multipliers and the
Hessian matrix, it is easy to show that, fixing the value of $s$ and
considering $p_i$ as variables, the value of $f(s)$ is maximized when $p_i =
\frac{1}{4}$ for all $i$. This means that with $k$ fixed, The value of $s$ is
maximized when all $p_i$ are equal. In this case, $s$ is explicitly calculated
as $s = \frac{\log 4}{\log 4k}$.

Since the values of $p_i$ are fixed and are $0 < p_i < 1$, $f(s)$ is strictly
monotonically decreasing. Also, using that $k^s$ is strictly monotonically
increasing and $f(0) = 4 > k^0 = 1$ we conclude that the solution to $f(s) =
k^s$ verifies $s > 0$.
\end{proof}

\begin{corollary}
\label{cor:linear}
  The asymptotic space used by a DN-tree with parameters $k \ge 1$ and $t$, is
$T(N) \in \Theta \left( N^s \right)$, where $N$ is the number of extent
accesses and $s$ is a real number between $0$ and $1$.
\end{corollary}

\begin{proof}
$k \ge 1$ implies $\frac{\log 4}{\log 4k} \le 1$
\end{proof}

The size of the DN-tree is, in the worst case, sublinear in the number of
extent accesses with a reasonable parametrization ($k>1$), as shown in
Corollary~\ref{cor:linear}. Moreover, Theorem~\ref{lem:bounds} provides a
tighter bound by a deep analysis of the worst case for $s$, which happens when
the probability of all transitions recorded in the DN-tree is homogeneous. 

Table~\ref{tab:matrixsize} shows bounds found numerically for different values
of $k$ and $p_{\mathrm{max}} = \max_{i \in \left[ 1, 4\right]}{p_i}$ using
Theorem~\ref{the:DN-tree_size}. When $p_{\mathrm{max}} = 0.25$, the bound is
that given by Lemma~\ref{lem:bounds}. Other cases are, for example,
$p_{\mathrm{max}} = 0.9$ and $k = 2$, where the DN-tree has size $O
\left(N^{0.52} \right)$, while with $k = 8$ the size is $O \left( N^{0.32}
\right)$. This corresponds approximately to the square and third roots of $N$,
respectively. With one MiB of memory, these configurations store about
$10^{12}$ and $10^{18}$ extent accesses, respectively.

\begin{table}[t]
 \centering
  \begin{small}
  \begin{tabular}{ | l | c | c | c | c | }
\hline
k & 1.5 & 2 & 4 & 8 \rule{0pt}{2.7ex} \\
\hline
size &
$O \left( N^{0.77}        \right)$ \rule{0pt}{3ex} &
$O \left( N^{\frac{2}{3}} \right)$ \rule{0pt}{3ex} &
$O \left( N^{\frac{1}{2}} \right)$ \rule{0pt}{3ex} &
$O \left( N^{\frac{2}{5}} \right)$ \rule{0pt}{3ex} \\
\hline
  \end{tabular}
  \end{small}
  \caption{ Bounds for size of the DN-tree. }
\label{tab:matrixsize}
\end{table}

\begin{table}[t]
 \centering
  \begin{small}
  \begin{tabular}{ | l | l | l | l | l | }
\hline
$p_{\mathrm{max}}/k$ &  1.5 &    2 &   4  &    8 \\
\hline
                0.25 & 0.77 & 0.67 & 0.50 & 0.40 \\
                0.3  & 0.77 & 0.67 & 0.50 & 0.40 \\
                0.4  & 0.77 & 0.66 & 0.50 & 0.40 \\
                0.5  & 0.76 & 0.65 & 0.49 & 0.39 \\
                0.6  & 0.74 & 0.64 & 0.48 & 0.38 \\
                0.7  & 0.72 & 0.61 & 0.46 & 0.37 \\
                0.8  & 0.69 & 0.58 & 0.44 & 0.35 \\
                0.9  & 0.63 & 0.53 & 0.40 & 0.32 \\
\hline
  \end{tabular}
  \end{small}
  \caption{ Exponent for the size of the DN-tree. }
\label{tab:matrixsize2}
\end{table}

The configuration parameter $k$ allows fine control over the growth of the
DN-tree. Table~\ref{tab:matrixsize2} summarizes values of $s$ for different
values of $k$ and distributions of data. However, the growth of the DN-tree can
be controlled further. One of the assumptions in Section~\ref{sec:definition}
is that queries follow some common pattern and do not change radically over
time. This means that after some time, we can simply stop updating the DN-tree,
as the behavior of the queries has already been captured. This is especially
useful in environments with heavy loads and lots of data accesses, because both
the CPU and memory usage are reduced.

\subsection{Error Analysis}
\label{sub:error}

\begin{figure*}[t]
\centering
\subfigure[\{0.30, 0.25, 0.25, 0.20\}]{
  \epsfig{file=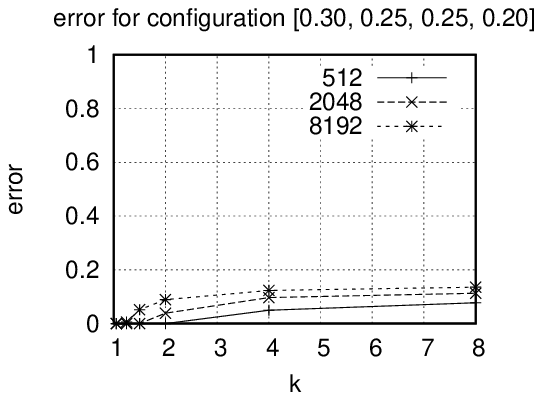, width=0.3\linewidth}
  \label{fig:error_a}
}
\subfigure[\{0.45, 0.25, 0.25, 0.05\}]{
  \epsfig{file=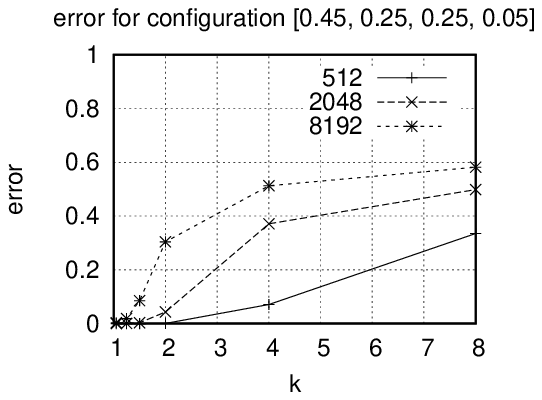, width=0.3\linewidth}
  \label{fig:error_b}
}
\subfigure[\{0.9, 0.09, 0.009, 0.001\}]{
  \epsfig{file=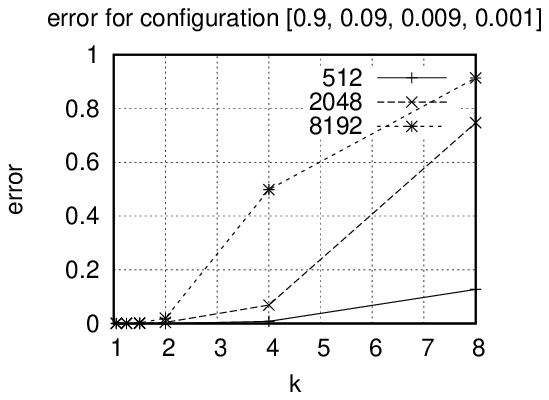, width=0.3\linewidth}
  \label{fig:error_c}
}
\caption{ Error for different configurations with $N=10^9$. }
\label{fig:error}
\end{figure*}

The compression used to store the matrix is lossy, meaning that we are not able
to recover exactly the original matrix but an approximation. In this section,
we give bounds for the error due to the compression.

We take the same assumptions as in Section~\ref{sub:size}: the matrix contents
follow an R-MAT pattern defined by the four probabilities $p_0$, $p_1$, $p_2$
and~$p_3$. Given a matrix $M$ and an approximation $\hat{M}$, we calculate the
error as the sum of the absolute value of all the elements of $M - \hat{M}$.

We have considered different configurations and generated matrices
corresponding to different values of the number of extent accesses, $N$. We
compare the matrix generated by our data structure, $\hat{M} = \left(
\hat{m}_{ij} \right)$, with the original matrix, $M = \left( m_{ij} \right)$,
and the error is computed as:

$$\mbox{error} =
\frac{1}{2 \left( N - 1 \right)}
\sum_{i,j} \left| m_{ij} - \hat{m}_{ij} \right|$$

Since the sum of the elements of both matrices is identical and equal to $N -
1$, the factor $2 \left( N - 1 \right)$ ensures that the error value is between
zero and one.

Figure~\ref{fig:error} shows graphs for three different configurations of the
values of $p_i$. The vertical axis corresponds to the error while the
horizontal axis represents the threshold growing constant $k$. Each graph has
three plots, corresponding to three different sizes of matrix $M$, $512 \times
512$, $2048 \times 2048$ and $8192 \times 8192$, representing datasets with
$512$, $2048$ and $8192$~extents.

The first graph, in Figure~\ref{fig:error_a}, corresponds to values of $p_i$
equal to \{$0.30$, $0.25$, $0.25$, $0.20$\}, which is very close to the uniform
distribution $p_i = 0.25$. With this distribution, the error is small for all
values of $k$. The graph in Figure~\ref{fig:error_c} corresponds to values of
$p_i$ equal to \{$0.9$, $0.09$, $0.009$, $0.001$\}, which is a very different
distribution and the error is higher. Figure~\ref{fig:error_b} corresponds to a
distribution between the other two, and the error is significantly higher than
the one in Figure~\ref{fig:error_a}, but does not grow as high as the one in
Figure~\ref{fig:error_c}.

We observe that for all configurations, a small value of $k$ is associated with
a small error. This is in contrast to the size analysis, where small values of
$k$ are associated with larger sizes of the data structure. For our
experiments, we choose a value of $k = 1.5$, which provides a very small error,
and experiments show that the data structure uses less than $0.2\%$ of the
available memory of the system and it represents less than $0.002\%$ of the
size of the full matrix $M$.

\section{Experimental Evaluation}
\label{sec:experiments}

We ran experiments with the DYDAP system using two different datasets. The
first dataset is a synthetic R-MAT graph~\cite{ChakrabartiZF04}, while the
second one is a graph built with information from Twitter. The first dataset is
used with a query that accesses all vertices and edges of the graph, and is
used to measure the throughput of the system. The Twitter dataset is used with
queries that generally access a small fraction of the database, and is useful
to evaluate the average response time of the system, as it is necessary that
these types of queries execute in a very short time.

Our proposal is compared to a state-of-the-art static partitioning method, used
in ParallelGDB~\cite{barguno2011parallelgdb}, which we use as baseline. This
method statically defines a fixed distribution function that does not adapt to
incoming queries.

\subsection{Prototype Implementation}
\label{sub:prototype}

In the experiments, we use an implementation of the complete system as
described in this paper, using a modified version of
DEX~4.2~\cite{martinez2012efficient} as the graph runtime. DEX is used to
execute the queries at each node. The system supports an arbitrary number of
nodes, and uses MPI for data communication between them. The distributed file
system is the same used by ParallelGDB, described
in~\cite{barguno2011parallelgdb}.

The execution of a query follows the model described in
Section~\ref{sec:definition}. Once the query is finished, there is a final
round of network communication where all nodes send information to the node
that started the query, which outputs the result.

%
%
%

\subsection{Synthetic Data}

The synthetic dataset is used to measure throughput, as well as the standard
deviation of the load and network communication of the nodes. The query
executed accesses the whole graph, and is used to simulate analytic scenarios
where queries explore a large portion of the database.

The throughput is measured in traversed edges per second (TEPS.) This unit
describes the amount of graph edges that the system processes each second while
executing a query.

We also report the standard deviation of the amount of load and network
communication. Since the standard deviation has the same units as the data,
its units are the number of traversed vertices and bytes, respectively.

\subsubsection*{Setup}

The graph used is the biggest connected component of an R-MAT
graph~\cite{ChakrabartiZF04} with $2^{26}$~vertices and mean degree $16$. The
resulting graph has more than $37$ million vertices and one billion edges
($37,165,451$ and $1,051,580,953$, respectively).

The query executed on the R-MAT graph is a BFS starting from a randomly
selected fixed vertex. The query accesses the whole graph, and exhibits poor
locality, as data is not accessed repeatedly. It is executed using $10$~BSP
phases. Each phase calculates the next hop using as a source the new vertices
explored in the previous phase, and the load and network communication is
recorded at each node.

The BFS query is executed in a cluster with identical nodes, with 2x AMD
Opteron 6164 HE CPUs running at $1.7$~GHz and with $48$~GiB of RAM. The
nodes run a Linux base OS with kernel version $2.6.32$. The graph engine is
configured to use a maximum of $16$~GiB of RAM and to run using only one core
of the $12$~available.

Additionally, the memory of the nodes is cleared before the first execution of
each configuration. The BFS query is executed twice, without clearing the
memory between the two executions. The results are reported as
$\mathrm{static}_1$, $\mathrm{static}_2$, $\mathrm{DYDAP}_1$ and
$\mathrm{DYDAP}_2$ where the subindex indicates whether it corresponds to the
first or second execution.

With this setup, the DN-tree uses a maximum of $27.2$~MiB on each node, which
is $0.16\%$ of available memory.

\subsection*{Throughput}

\begin{figure}[t]
\centering
\epsfig{file=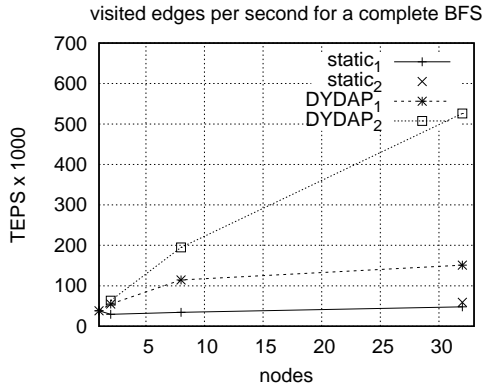, width=0.8\linewidth}
\caption{ Throughput of a BFS. }
\label{fig:throughput}
\end{figure}

\begin{table}[t]
 \centering
  \begin{small}
  \begin{tabular}{ | l | r r r r | }
    \hline
system / \#nodes &      1 &      2 &       8 &      32 \\
\hline
(1) static$_1$   & 38,180 & 29,425 &  34,427 &  48,031 \\
(2) static$_2$   &      - &      - &       - &  59,061 \\
(3) DYDAP$_1$    & 38,180 & 54,881 & 114,389 & 151,198 \\
(4) DYDAP$_2$    &      - & 63,310 & 194,882 & 525,790 \\
\hline
Speedup (3)/(1)  &      1 &    1.9 &     3.3 &     3.1 \\
Speedup (4)/(3)  &      - &    1.2 &     1.7 &     3.5 \\
Speedup (4)/(2)  &      - &      - &       - &     8.9 \\
Speedup (4)/(1)  &      - &    2.3 &     5.6 &    10.9 \\
    \hline
  \end{tabular}
  \end{small}
 \caption{ Throughput in TEPS and speedups. }
 \label{tab:time}
\end{table}

Table~\ref{tab:time} and Figure~\ref{fig:throughput} show the throughput in
TEPS of the BFS query both for the static method and our proposal in clusters
with $2$, $8$ and $32$~nodes. The graph also shows a second execution of the
query in DYDAP for all configurations and a second execution of the static
method with $32$~nodes. The differences in the first execution are due to the
load balancing, and the improvements of the second executions are due to the
use of the cache of the computers.

With one node, the result is independent of the system, as no distribution
takes place, and corresponds to executing the regular graph manager; in this
case the regular version of DEX, modified to communicate with the partition
manager.

When executed using two nodes, the baseline explores $29,425$~edges per second.
The execution of the query using the static method is used to collect data to
partition the database according to this query load and generate the first
distribution function for our method. When the query is executed using DYDAP,
it achieves $54,881$~TEPS, or $87\%$~better throughput.

Using eight nodes, the baseline traverses $34,427$~edges per second, while
DYDAP explores $114,389$. In this case, the throughput of DYDAP is
$3.3$~times that of the baseline, and more than twice the throughput of DYDAP
with two nodes. A second execution achieves a throughput of $194,882$~TEPS,
which is $70\%$ more than the throughput of the first execution.

With $32$~nodes, the query is executed two times on each system. The static
method explores $48,031$ edges per second during the first execution and
$59,061$ during the second one. This improvement of $23\%$ in throughput is due
to the cache specialization. Our system performs its first execution with a
throughput of $151,198$~TEPS, which is $3.15$~times the throughput of the
first execution with the static method. The second execution achieves
$525,790$~TEPS, which is almost an order of magnitude faster than the baseline.

\subsection*{Load}

The results of the following two subsections show how our system balances the
load and network communication between nodes, allowing the improvement seen in
the first execution of the query. The further improvement on the second
execution over the first one shows that our system is better at doing the cache
specialization of the nodes.

\begin{figure*}
\centering
\subfigure[$2$ nodes.]{
  \epsfig{file=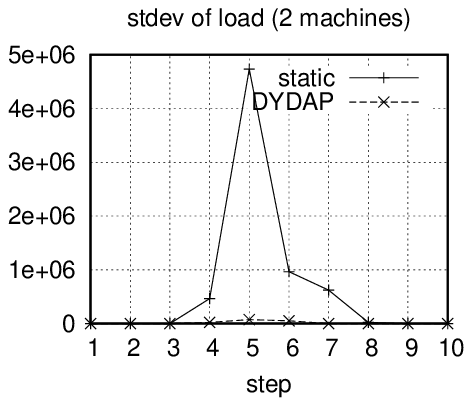, width=0.3\linewidth}
  \label{fig:load-2b}
}
\subfigure[$8$ nodes.]{
  \epsfig{file=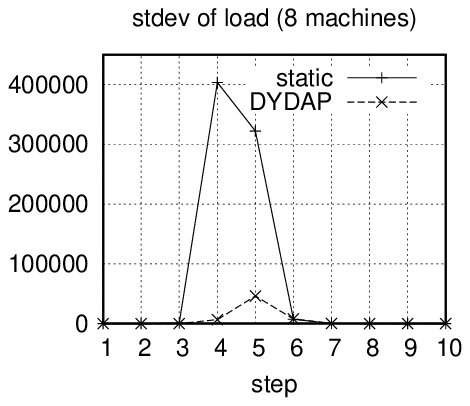, width=0.3\linewidth}
  \label{fig:load-8b}
}
\subfigure[$32$ nodes.]{
  \epsfig{file=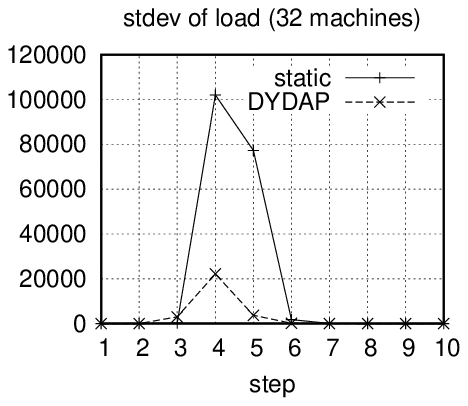, width=0.3\linewidth}
  \label{fig:load-32b}
}
\subfigure[$2$ nodes.]{
  \epsfig{file=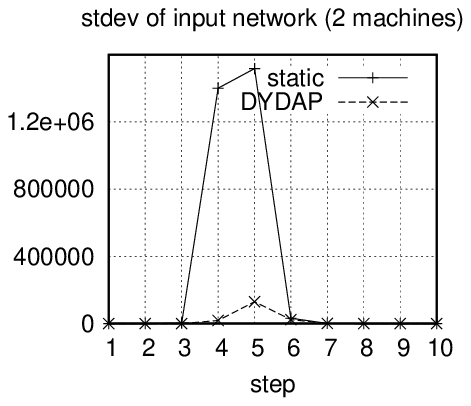, width=0.3\linewidth}
  \label{fig:net-2b}
}
\subfigure[$8$ nodes.]{
  \epsfig{file=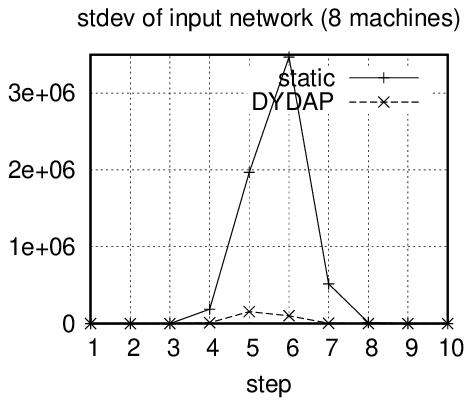, width=0.3\linewidth}
  \label{fig:net-8b}
}
\subfigure[$32$ nodes.]{
  \epsfig{file=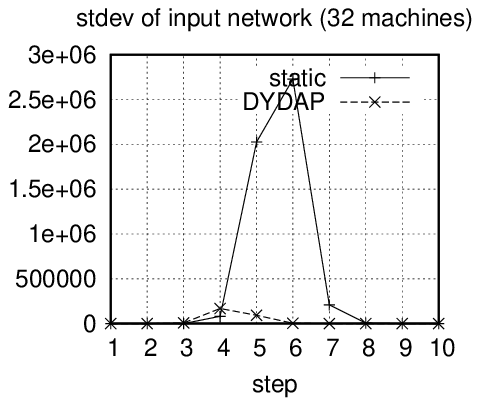, width=0.3\linewidth}
  \label{fig:net-32b}
}
\caption{ Standard deviation of load (top) and network communication (bottom)
at each phase. }
\label{fig:load_network}
\end{figure*}

For the configuration with two nodes, Figure~\ref{fig:load-2b} shows the
standard deviation of the load of both nodes at each phase, measured as the
number of vertices traversed, using the static method and our proposal. Notice
how at each step, the load at each node is almost identical using our proposal,
while there are differences using the static method. This differences imply
that one node is idly waiting for the other to finish, in order to start the
network communication. This accounts for the time saved when using our system.

Figures~\ref{fig:load-8b} and~\ref{fig:load-32b} show the standard deviation of
the load when using $8$ and $32$~nodes for both methods. Here we observe how
the standard deviation is much lower for our method, which means that all nodes
have a similar load, drastically reducing the time needed to finish the query.

\subsection*{Network Communication}

Using two nodes, Figure~\ref{fig:net-2b} shows the standard deviation of
incoming network communication between each pair of consecutive phases for one
of the two nodes. Since there are only two nodes in the system, it is
equivalent to the standard deviation of outgoing network communication. We
observe that the distribution of network communication is uneven with the
static method, but almost equal in DYDAP. This is a side effect of having a
balanced load.

Figures~\ref{fig:net-8b} and~\ref{fig:net-32b} show the standard deviation of
incoming network communication for the systems with $8$ and $32$ nodes. In
these cases it is also clear that network communication is much more balanced
in our proposal than in the baseline, which means that it completes much faster.

\subsection{Real Data}

The query on the Twitter dataset is used to compute the average response time
of the system. The query executed on this dataset explores a small portion of
the graph, and simulates scenarios such as web servers, where there are a large
number of queries and each query does not access a large fraction of the
database. The average response time of the system is measured in seconds, and
is calculated as the duration of a time frame divided by the number of queries
that the system executes during that time frame.

\subsubsection*{Setup}

The Twitter dataset is generated by combining the data obtained
from~\cite{twitterdataset1} and~\cite{twitterdataset2}, it contains over
$40$~million users, $26$~million tweets and one billion following/follower
relationships. This graph has four different types of vertices: \emph{user},
\emph{tweet}, \emph{hashtag} and \emph{url}; and seven types of edges:
\emph{tweets}, \emph{follows}, \emph{receives}, \emph{depicts}, \emph{retweet},
\emph{tags} and \emph{reference}. Each edge type joins two specific vertex
types, for example, a \emph{tweets} edge joins a \emph{user} vertex with a
\emph{tweet} vertex. The database with all the information has a size of
248~GiB.

The query executed on this dataset is a 2-hop. It simulates the generation of
the home page of a user. Given $u$, a vertex of type \emph{User}, the first hop
explores the edges of type \emph{follows} that start at $u$. After this step,
the query has a set of vertices of type \emph{User} that $u$ follows. The
second hop expands this set through edges of type \emph{tweets}. This second
hop is done through a different edge type; this means that a different data
structure is accessed.
The final results are the tweets of the users that $u$ follows. This is an
example of a system issuing queries interactively to generate web pages in real
time. Thus, it is very important to have low response times.

The 2-hop query has been executed on Amazon~EC2. The instances used are of type
m2.xlarge: $17.1$~GiB of memory and $6.5$~EC2 Compute Units ($2$~virtual cores
with $3.25$~EC2 Compute Units each). The database is stored using an EBS volume
for each instance. Also, the graph engine is configured to use a maximum of
$12$~GiB of memory.

With this setup, the DN-tree uses a maximum of $7.5$~MiB on each node, which is
$0.06\%$ of available memory.

\subsection*{Average response time}

On each system, the query is executed multiple times using a random vertex of
type \emph{User} as a parameter. Each node has $10$~threads executing queries
concurrently during $90$~minutes. The results reported are the average response
time of the system for the last $30$~minutes, which is calculated as the time
elapsed divided by the number of queries executed. The whole experiment is
executed three times and the best results for each configuration are reported.

\begin{figure}[t]
\centering
\epsfig{file=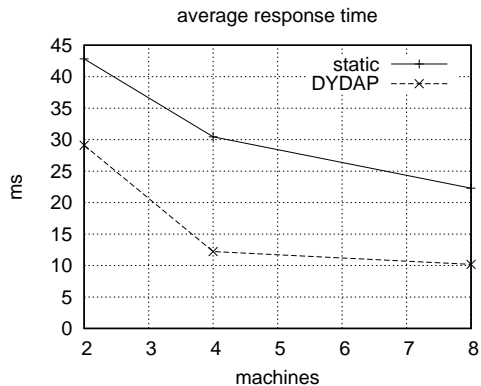, width=0.8\linewidth}
\caption{ Average response time for 2-hops during the last $30$~minutes of
execution. }
\label{fig:avg}
\end{figure}

Figure~\ref{fig:avg} shows the average response time for both methods for
configurations with $2$, $4$ and $8$~nodes. With $2$~nodes, the average
response time of our system is $67.9\%$ of that of the baseline. For $4$ and
$8$~nodes, our average response time is reduced to $40.1\%$ and $45.6\%$ with
respect to the baseline.

We observe that, with four nodes, the average response time is divided by two.
Using eight nodes does not improve the response time by a large factor because,
on average, the queries do not require a lot of computation, which means that
the network costs are a large fraction of the total cost.

\subsection{Summary}

The experiments performed using synthetic data show that our system scales with
the number of nodes, and our dynamic data partitioning outperforms a state of
the art static method by a factor of three during the first execution. During
the second execution, when the system is warmed up, the baseline barely
improves, while DYDAP achieves a throughput an order of magnitude larger than
that of the baseline. This is because our system is better at performing node
cache specialization. Also, analyzing the load and network usage of each node,
we see that it is much more balanced in DYDAP, which allows a much better
performance.

The experiments with real data show that DYDAP has an average response time
that is less than half of the baseline when using at least four nodes. We
see that with four nodes, DYDAP already achieves a very good result, with
an average response time of $12$~ms. With eight nodes, the average response
time is similar, $10$~ms. This is due to the fact that this query is relatively
simple, so the time spent in network communication during the query dominates
with respect to the time spent accessing the data to answer the query when the
number of nodes is large. Our system uses all the available nodes to solve each
query. We leave as future work to analyze the compromise between the resources
dedicated to solve each query and the response time of the system.


\section{Related Work}
\label{sec:related}

There are descriptions of several distributed systems in the literature. The
majority of them work with relational or key-value models, but some of them
employ a graph model. In this section, we review some of the distributed systems
that work with graphs and also discuss generic approaches that are applied both
to graph and non-graph models.

MapReduce has become in the latest years the main paradigm for processing large
batches of embarrassingly parallel tasks. Graph mining tasks such as calculating
the diameter of the graph or the connected components can be computed using
MapReduce. For example, Pegasus~\cite{kang2011pegasus} is a library that
represents the graphs as a sparse matrices, which are loaded in a MapReduce
platform. These matrices can be manipulated with sequences of algebraic
operations that simulate graph mining operations such as calculating the
diameter of a graph or the connected components. However, these systems do not
take into account locality decisions and are not suitable for obtaining fast
results from queries because of the overheads associated to start MapReduce
tasks. In~\cite{cohen2009graph}, the authors present an analysis of MapReduce
techniques to perform graph operations and conclude that MapReduce has severe
limitations. The communication between processes is a bottleneck, as the
programs need to exchange a large amount of data. Although some limitations are
inherent to the nature of graph data, the solution presented in our paper could
be implemented in MapReduce platforms to improve its efficiency.

Google proposed a framework for efficient graph processing called
Pregel~\cite{malewicz2010pregel} as an alternative to MapReduce. Pregel defines
a message communication network among the vertices that uses the edges as
connections. Graph algorithms are programmed as sequences of message exchanges
in this network. Other papers by different authors improve Pregel. For example,
Sedge~\cite{yang2012towards} adds the management of different partitions over
the original graph, and in~\cite{shangcatch} the authors propose a method that
changes the partitioning by exploring the behavior of the system during a
working window. PowerGraph~\cite{gonzalez2012powergraph} is a distributed
system using a model very similar to that of Pregel, but partitions vertices
instead of edges. Although Pregel alleviates the network bottlenecks of
MapReduce, this solution is also oriented to large scale graph mining and not
queries with short response time. Besides, it requires rewriting all graph
algorithms in terms of such network communication model, which is not close to
the traditional graph abstract data type. In contrast, we use a traditional API
programming model that allows the use of classic implementations of graph
algorithms.

One method used to improve system throughput is cache specialization, which
ensures data locality. In~\cite{elnikety2007tashkent}, the authors present a
load balancing method for relational databases that assigns transactions to
replicas in a way that they are executed in memory, thus reducing disk usage.
This method is an example of system using data locality. Our system is
different because it does not use replicas, but rather distributes the only
copy of the data between nodes, also exploiting data locality to reduce disk
usage.

In~\cite{pujol2012little}, the authors propose a method for partitioning and
replication of social networks that minimizes the number of replicas necessary
to guarantee data locality for all vertices. A generalization of this method is
presented in~\cite{mondal2012managing}. These two methods focus on controlling
the number of replicas to minimize network communication. Our method does not
replicate data and additionally balances the load of the computers.
Also, these methods require an expensive process of recomputing partitionings
when the graph is modified, while our partitioning recomputations are very
fast.

ParallelGDB~\cite{barguno2011parallelgdb} is a static method to distribute a
graph database that works by specializing the caches of the nodes in the
system, reducing disk and network I/O. Our method works similarly, but the
partitioning is adapted to the incoming queries, and provides high performance
also in cases where the workload is skewed.

Schism~\cite{curino2010schism,tatarowicz2012lookup} is an approach to
relational database partitioning and replication that uses a method similar to
ours to partition data. The method constructs a graph that assigns one vertex
to each tuple of the database. The resulting graph is often very large and it
has to be sampled to improve efficiency, while our system does not use
sampling. A limitation of this method is that it needs to know the query
workload before execution starts, and what data each query accesses, instead of
adapting to incoming queries.

\section{Conclusions}
\label{sec:conclusions}

In this paper, we proposed a distributed system design in two levels: secondary
storage and memory manager. Independently of the storage implementation used,
the memory manager specializes the cache of the computing nodes. The
memory manager uses a new method to dynamically partition data in a
graph database. This data partitioning balances the load and minimizes the
amount of network communication in the distributed system. The goal is to adapt
the data partitioning to the incoming query workload, and this goal is achieved
as shown by the experiments.

Our experimental results show that our distributed system works in a variety of
different situations. It provides great performance when used to analyze huge
amounts of data, and also in a web environment, where queries are simpler but
require a very short response time.


\section*{Acknowledgments}

The members of DAMA-UPC thank the Ministry of Science and Innovation of Spain
and Generalitat de Catalunya, for grant numbers TIN2009-14560-C03-03 and
SGR-1187 respectively.
Experiments presented in this paper were carried out using the Grid'5000
experimental testbed, being developed under the INRIA ALADDIN development
action with support from CNRS, RENATER and several Universities as well as
other funding bodies (see https://www.grid5000.fr).

\begin{small}

\bibliographystyle{IEEEtran}
\bibliography{ondemand}

\end{small}

\appendix

\section{Akra-Bazzi theorem}
\label{app:akra-bazzi_theorem}
\begin{theorem}[Akra-Bazzi~\cite{akra1998solution}]
 Given a recurrence of the form:
$$T(n) = g(n) + \sum_{i=0}^c a_i T(b_i n + h_i(n))$$
with the following conditions:
 $a_i$ and $b_i$ are constants;
 $a_i > 0$;
 $0 < b_i < 1$;
 $|g(n)| \in O\left(n^d\right)$, where $d$ is a constant;
 and, $|h(n)| \in O\left(\sfrac{n}{\left( \log n \right)^2}\right)$.

The asymptotic behavior of $T(n)$ is given by:
$$T(n)\in\Theta\left(n^s\left(1+\int_1^n\frac{g(u)}{u^{s+1}}du\right)\right)$$
where $s$ is the solution to the equation:
$$\sum_{i=0}^{c} a_i b_i^s = 1$$
\end{theorem} 

\end{document}